%% file: root.tex
\documentclass[letterpaper, 10 pt, conference]{ieeeconf}  
\IEEEoverridecommandlockouts                              
\input{preamble}

\title{\LARGE \textbf{
A Layered Control Perspective on Legged Locomotion: \\
Embedding Reduced Order Models via Hybrid Zero Dynamics}
}

\author{Sergio A. Esteban, Max H. Cohen, Adrian B. Ghansah and Aaron D. Ames
\thanks{
This research was supported by NSF award No. 1932091 and Technology Innovation Institute (TII).
}
\thanks{
The authors are with the Division of Engineering and Applied Science, California Institute of Technology, Pasadena, CA 91125 USA, {\tt\small\{sesteban, maxcohen, aghansah, ames\}@caltech.edu}.
}
}

\begin{document}

\maketitle
\thispagestyle{empty}
\pagestyle{empty}

\input{Sections/Abstract}

\input{Sections/Introduction}

\input{Sections/Preliminaries}
\input{Sections/HZDROM}


\input{Sections/BipedApplication}

\input{Sections/Conclusion}


\bibliographystyle{ieeetr}
\bibliography{References/refs}

\end{document}

%% file: preamble.tex
\usepackage{amsmath}
\usepackage{amssymb}
\usepackage{amsthm}
\usepackage{mathtools}
\usepackage{derivative}
\usepackage{xcolor}
\usepackage{bm}
\usepackage[font=small, labelfont=bf]{caption}

\usepackage[noadjust]{cite}
\usepackage{url}

\usepackage{graphicx}
\usepackage[export]{adjustbox} 
\graphicspath{{figures/}}
\usepackage{subfig} 

\usepackage{balance}

\usepackage{hyperref} 
\hypersetup{
    colorlinks=true,
    linkcolor=black,
    urlcolor=cyan,
}

\newtheorem{theorem}{Theorem}
\newtheorem{lemma}{Lemma}

\theoremstyle{definition}
\newtheorem{definition}{Definition}

\newtheorem{remark}{Remark}

\newcommand{\R}{\mathbb{R}}

\newcommand{\mc}[1]{\mathcal{#1}}
\newcommand{\T}{^\top}

\newcommand{\rank}{\operatorname{rank}}

\newcommand{\bzero}{\mathbf{0}}

\newcommand{\bd}{\mathbf{d}}

\renewcommand{\bf}{\mathbf{f}} 
\newcommand{\bg}{\mathbf{g}}
\newcommand{\bh}{\mathbf{h}}

\newcommand{\bk}{\mathbf{k}}

\newcommand{\bq}{\mathbf{q}}
\newcommand{\br}{\mathbf{r}}

\newcommand{\bu}{\mathbf{u}}

\newcommand{\bx}{\mathbf{x}}

\newcommand{\bz}{\mathbf{z}}

\newcommand{\bA}{\mathbf{A}}
\newcommand{\bB}{\mathbf{B}}

\newcommand{\bD}{\mathbf{D}}

\newcommand{\bF}{\mathbf{F}}
\newcommand{\bG}{\mathbf{G}}
\newcommand{\bH}{\mathbf{H}}
\newcommand{\bI}{\mathbf{I}}

\newcommand{\bK}{\mathbf{K}}

\newcommand{\bN}{\mathbf{N}}

\newcommand{\bP}{\mathbf{P}}
\newcommand{\bQ}{\mathbf{Q}}

\newcommand{\bpsi}{\bm{\psi}}

\newcommand{\bPhi}{\bm{\Phi}}

\newcommand{\bXi}{\bm{\Xi}}
\newcommand{\bOmega}{\bm{\Omega}}

\newcommand{\flow}{\bm{\varphi}}

\usepackage{xcolor}
\definecolor{myblue}{RGB}{49, 114, 174}
\definecolor{myred}{rgb}{0.796, 0.235, 0.2}
\definecolor{mygreen}{rgb}{0.22, 0.596, 0.149}
\definecolor{mypurple}{rgb}{0.584,0.345,0.698}

\usepackage{blindtext}

%% file: Sections/Abstract.tex
\begin{abstract}
Reduced-order models (ROMs) provide a powerful means of synthesizing dynamic walking gaits on legged robots. Yet this approach lacks the formal guarantees enjoyed by methods that utilize the full-order model (FOM) for gait synthesis, e.g., hybrid zero dynamics. This paper aims to unify these approaches through a layered control perspective. In particular, we establish conditions on when a ROM of locomotion yields stable walking on the full-order hybrid dynamics. To achieve this result, given an ROM we synthesize a zero dynamics manifold encoding the behavior of the ROM---controllers can be synthesized that drive the FOM to this surface, yielding hybrid zero dynamics. We prove that a stable periodic orbit in the ROM implies an input-to-state stable periodic orbit of the FOM's hybrid zero dynamics, and hence the FOM dynamics. This result is demonstrated in simulation on a linear inverted pendulum ROM and a 5-link planar walking FOM. 
\end{abstract}

%% file: Sections/Introduction.tex
\section{Introduction}
%
From the influential Raibert heuristic \cite{raibert1986legged} to the rich catalog of inverted pendulum variants \cite{wensing2017template}, reduced-order models (ROMs) have long been used to achieve locomotion on legged robots. The success of these models is attributed to the fact that
they generalize well across most legged platforms and, by design, distill the complexity of whole-body dynamics down to a lower-dimensional description while retaining key features that govern locomotion. Models such as the linear inverted pendulum (LIP)\cite{kajita20013d}, and its variants \cite{xiong20223, gong2020angular, dai2024multi}, and the spring-loaded inverted pendulum (SLIP) \cite{blickhan1989spring, wensing2013high} focus on regulating the center of mass and center of pressure to enable dynamic behavior such as walking or running. At their core, these ROMs provide a framework for guiding the robot's actuated dynamics in response to its unactuated dynamics. In this work, we develop a formal theory for how these ROMs exploit the structure of uncontrollable dynamics, termed the \textit{zero dynamics}, to enable full-order model stability.

ROMs have recently been incorporated into broader control frameworks. Applications include foot step planning \cite{pratt2006capture, englsberger2011bipedal, englsberger2015three, mesesan2019dynamic}, reinforcement learning \cite{green2021learning, bang2024rl}, model predictive control \cite{esteban2025reduced, gibson2022terrain}, and nominal gait stabilization \cite{ghansah2024dynamic, ye2023decoupled}. Their significance in control frameworks has also led to synthesizing ROMs using learning \cite{csomay2024robust, chen2024reinforcement, da2019combining}. Despite their long history and widespread usage, theoretical frameworks connecting ROMs and full-order models (FOMs) remain underexplored in the literature.
\begin{figure}
    \centering
    \includegraphics[width=0.8\linewidth]{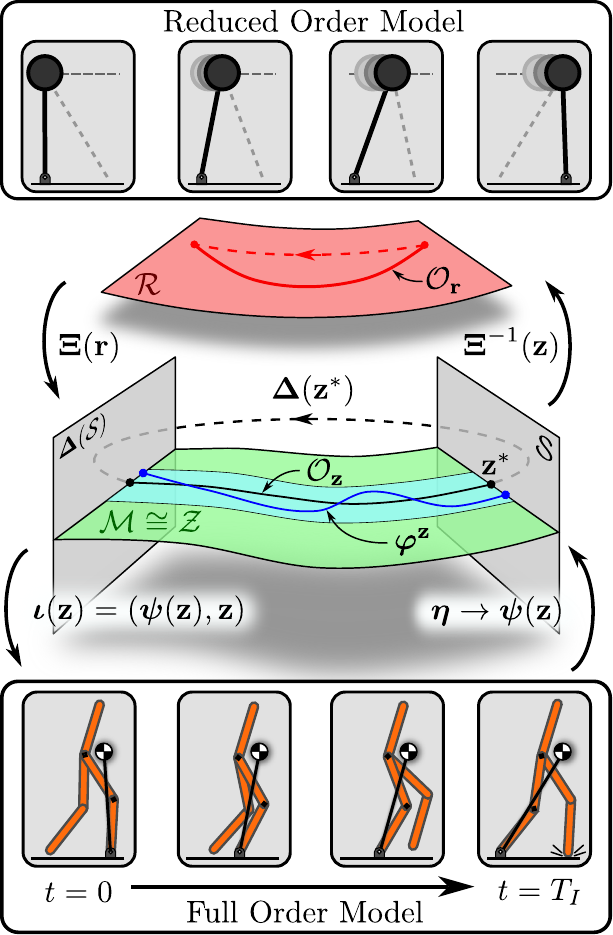}
    \caption{We consider a layered control approach that leverages a reduced-order model (ROM) of locomotion to stabilize the dynamics of a full-order model (FOM) bipedal robot.}
    \label{fig:layered_architecture}
        \vspace{-7.0mm}
\end{figure}
One tool for analyzing the connection between the FOM dynamics and their corresponding ROM abstractions is the approximate simulation approach \cite{girard2009hierarchical} wherein a Lyapunov-like function, termed a \textit{simulation function}, provides a quantitative metric characterizing the distance between trajectories of ROMs and FOMs. In these approaches, formal characterizations are developed for hierarchies composed of a class of linear systems. The work presented in \cite{kurtz2019formal, kurtz2020approximate} is based on approximate simulation and provides a formal treatment of a scenario that considers a bipedal robot and the canonical LIP model. 

The hybrid zero dynamics (HZD) framework provides an alternative approach to generating stable walking by leveraging the FOM hybrid dynamics \cite{westervelt2003hybrid}.  HZD has a well-established track record for gait generation of legged robots \cite{hereid2014dynamic, ames2017hybrid, reher2016realizing, reher2019dynamic}. This approach focuses on designing \textit{virtual constraints} for the actuated dynamics that, by construction, enforce the unactuated dynamics to evolve on a lower-dimensional, attractive, and invariant manifold -- even in the presence of impulsive effects, which are inherent in legged locomotion. These virtual constraints are typically synthesized using large-scale offline trajectory optimization to produce a set of outputs that, when tracked, result in stabilization of the hybrid system. In the HZD framework, the zero dynamics provide an \emph{exact} reduced-order representation of the FOM. However, these exact ROMs, in general, cannot be represented in terms of the aforementioned LIP models, which may be seen as an approximation of the corresponding zero dynamics \cite{grizzle2017virtual}. Although the HZD framework has been successfully realized on multiple systems, its reliance on offline trajectory optimization of periodic orbits limits its suitability for online control.



In this paper, we present a formal framework for achieving stable locomotion for full-order hybrid systems that unites elements of ROMs with the HZD framework. To this end, we take inspiration from the framework of \emph{zero dynamics policies} (ZDPs) \cite{csomay2024robust, compton2024constructive}, a constructive method for stabilizing hybrid systems that drives the system to a lower-dimensional invariant manifold encoding a ROM.  In particular, we encode behaviors of a desired ROM within a zero dynamics manifold embedded in the full-order state space. We establish conditions under which this manifold is invariant for the full-order hybrid dynamics and produces hybrid zero dynamics (Sec. \ref{sec:HZD}). Rather than assuming exponential stability of these zero dynamics (as in the HZD framework) or learning a ZDP, we formally characterize the connection between the zero dynamics and a ROM using an input-to-state stability (ISS) framework. In particular, we show that, under certain conditions, exponentially stable periodic orbits of the ROM correspond to exponentially ISS (E-ISS) periodic orbits of the zero dynamics (Sec. \ref{sec:ROM}). Using this connection, the main result of this paper (Theorem \ref{theorem:main}) establishes that exponentially stable periodic orbits of the ROM produce E-ISS periodic orbits of the full-order hybrid dynamics. We elucidate our main theoretical results by applying them in a case study involving a 5-link planar biped and a variant of the canonical LIP (Sec. \ref{sec:biped}).

Compared to traditional HZD approaches \cite{westervelt2003hybrid,hereid2014dynamic, ames2017hybrid, reher2016realizing, reher2019dynamic}, our approach avoids computationally expensive offline trajectory optimization by directly generating a gait from a computationally efficient ROM. The limitation of our approach, compared to HZD, is that, in general, it produces suboptimal gaits that do not fully exploit the FOM dynamics. 





%% file: Sections/Preliminaries.tex
\section{Preliminaries} \label{sec:preliminaries}
\textbf{Notation:}
Given a manifold $\mc{M}$, we use $\mathsf{T}_{\bx}\mc{M}$ to denote the tangent space at $\bx\in\mc{M}$ and $\mathsf{T}\mathcal{M}$ to denote the tangent bundle. We write $\mc{M}\cong\mc{N}$ to denote that two sets $\mc{M}$ and $\mc{N}$ are diffeomorphic in that there exists of a diffeomorphism $\bPhi\,:\,\mc{M}\rightarrow\mc{N}$ (a smooth, invertible, bijective mapping, with smooth inverse) from $\mc{M}$ to $\mc{N}$. We use $\mc{B}_{\delta}(\bx^*)$ to denote an open ball of radius $\delta>0$ centered at $\bx^*$ and $\overline{\mc{B}_{\delta}(\bx^*)}$ to denote its closure. The Euclidean norm is denoted by $\|\cdot\|$ and $\|\cdot\|_{\mc{M}}$ denotes the point-to-set distance for a set $\mc{M}$. A manifold $\mc{M}$ is said to be invariant for an ODE $\dot{\bx}=\bf(\bx)$ defined by the vector field $\bf\,:\,\mc{M}\rightarrow\mathsf{T}\mc{M}$ if its solutions $t\mapsto \bx(t)$ satisfy $\bx(t)\in\mc{M}$ for all $t\in\R$. Similarly, a manifold $\mc{M}$ is said to be invariant for a difference equation $\bx^+=\bF(\bx)$ if $\bF(\bx)\in\mc{M}$ for all $\bx\in\mc{M}$. A continuous function $\alpha\,:\,\R_{\geq0}\rightarrow\R_{\geq0}$ is said to be a class $\mc{K}$ function ($\alpha\in\mc{K}$) if $\alpha(0)=0$ and $\alpha$ is strictly increasing.

\subsection{Full-Order Models and Hybrid Dynamics} \label{sec:fom_dynamics}
In this paper, we consider legged robots modeled as a hybrid system with impulse effects, which we will refer to as the \emph{full-order model} (FOM) of the system under consideration. 
These systems have configuration
$\bq \in \mc{Q} \subseteq \R^n$ and state $\bx = (\bq, \dot{\bq}) \in \mc{X} \coloneqq \mathsf{T}\mc{Q} \subseteq \R^{2n}$. Using the Euler-Lagrange equations, we define the continuous dynamics as:
\begin{equation} \label{eq:robotics_eq}
    \bD(\bq)\ddot{\bq} + \bH(\bq,\dot{\bq}) = \bB\bu,
\end{equation}
where $\bD : \mc{Q} \rightarrow \R^{n\times n}$ is the positive definite mass-inertia matrix, $\bH : \mc{X} \rightarrow \R^n$ contains centrifugal, Coriolis, and gravitational terms, $\bB \in \R^{n \times m}$ is the actuation matrix, and $\bu \in \mc{U} \subseteq \R^m$ is the control input. In general, legged robots have unactuated coordinates in that $\mathrm{rank}(\bB) = m < n$, making the dynamics underactuated.
Furthermore, we may write these dynamics in state space form as:
\begin{equation}\label{eq:fom_continuous_dyn}
    \dot{\bx} = 
    \underbrace{
    \begin{bmatrix}
        \dot \bq \\
        -\mathbf{D}^{-1}(\bq) \mathbf{H}(\bq, \dot \bq)
    \end{bmatrix}
    }_{\bF(\bx)}
    +
    \underbrace{
    \begin{bmatrix}
        \mathbf{0} \\
        \mathbf{D}^{-1}(\bq) \mathbf{B}
    \end{bmatrix}
    }_{\bG(\bx)}
    \bu,
\end{equation}
where the drift 
$\bF\,:\,\mc{X} \rightarrow \R^{2n}$
and actuation 
$\bG : \mc{X} \rightarrow \R^{2n\times m}$ 
terms are 
continuously differentiable. Additionally, under an appropriate switching surface $\mc{S} \subset \mc{X}$, typically defined by a foot striking the ground, the dynamics undergo discrete transitions which are represented by:
\begin{equation} \label{eq:fom_reset_map}
    \bx^+ = \boldsymbol{\Delta}(\bx^-),
\end{equation}
where  $\bm{\Delta}\,:\, \mc{S} \rightarrow \mc{X}$ is the locally Lipschitz reset map,
$\bx^{-}\in\mc{S}$ denotes the state just before impact and $\bx^{+}\in\mc{X}\setminus\mc{S}$ denotes the state just after impact.
Combining \eqref{eq:fom_continuous_dyn} and \eqref{eq:fom_reset_map}, the dynamics can be written as a hybrid control system:
\begin{equation} \label{eq:fom_HC}
        \mathcal{HC}
        =
        \left\{
        \begin{array}{ll}
              \dot {\bx} = \bF(\bx) + \bG(\bx) \bu & \text{if} \quad \bx \in \mathcal{X} \backslash \mc{S}\\
              \bx^{+} = \bm{\Delta}(\bx^{-})    & \text{if} \quad \bx^{-} \in \mc{S}.
        \end{array} 
        \right. 
\end{equation}
Given a locally Lipschitz feedback controller $\bk\,:\,\mc{X}\rightarrow\mc{U}$ we obtain the closed-loop hybrid system:
\begin{equation} \label{eq:fom_H}
        \mathcal{H}
        =
        \left\{
        \begin{array}{ll}
              \dot {\bx} = \bF(\bx) + \bG(\bx) \bk(\bx) & \text{if} \quad \bx \in \mathcal{X} \backslash \mc{S}\\
              \bx^{+} = \bm{\Delta}(\bx^{-})    & \text{if} \quad \bx^{-} \in \mc{S}.
        \end{array} 
        \right. 
\end{equation}
Let the flow of these dynamics be denoted by $\bm{\varphi}_{t}(\bx)$, which returns the state reached at time $t$ under the continuous dynamics \eqref{eq:fom_continuous_dyn} when starting from state $\bx\in\mc{X}$. 
To further characterize the discrete evolution of \eqref{eq:fom_H}, define the time-to-impact function $T_{I}\,:\,\tilde{\mc{S}}\rightarrow\R_{>0}$ as:
\begin{equation}\label{eq:fom-time-to-impact}
    T_I(\bx) \coloneqq \inf\{t \geq 0\; | \; \boldsymbol\varphi_t( \mathbf\Delta(\bx)) \in \mc{S}\},
\end{equation}
where $\tilde{\mc{S}} \coloneqq \{\bx\in\mc{S}\,|\,T_{I}(\bx)\in(0,\infty) \}$,
which captures the time elapsed between discrete transitions.
The above formulation allows for casting the hybrid system \eqref{eq:fom_H} as a discrete-time system via the Poincar\'e map:
\begin{equation}\label{eq:fom_discrete_dyn}
    \bx_{k+1} = \bP(\bx_{k}),\, \bP(\bx) \coloneqq \bm{\varphi}_{T_{I}(\bx)}\big(\bm{\Delta}(\bx)\big),\, \bP\,:\,\tilde{\mc{S}}\rightarrow\mc{S}.
\end{equation}
As is common in the literature, we will often, with an abuse of notation, denote the Poincar\'e map as $\bP\,:\,\mc{S}\rightarrow\mc{S}$ with the understanding that $\bP$ is a partial function (see \cite{ames2014rapidly,westervelt2003hybrid}).
The Poincar\'e map provides a powerful tool for studying the stability of orbits of \eqref{eq:fom_H}. The flow of \eqref{eq:fom_H} is said to be periodic with period $T\geq0$ if there exists a state $\bx^*\in\mc{S}$ such that $\flow_T(\bm{\Delta}(\bx^*))=\bx^*$. This induces a periodic orbit:
\begin{equation}\label{eq:hybrid-periodic-orbit}
    \mc{O} \coloneqq \{ \flow_t(\bm{\Delta}(\bx^*))\,|\,t\in[0,T],\,T=T_{I}(\bx^*)\}.
\end{equation}
Importantly, stability of $\mc{O}$ is equivalent to that of $\bP$ \cite{westervelt2018feedback}.

\subsection{Actuation Decompositions} \label{sec:fom_diffeomorphism}
In this paper, we analyze the stability properties of hybrid systems by decomposing their dynamics into actuated and unactuated components. To this end, let $\bPhi\,:\,\mc{X}\rightarrow\mc{N}\times\mc{Z}$ be a diffeomorphism such that:
\begin{equation}
    \begin{bmatrix}
        \bm{\eta} \\ \bz
    \end{bmatrix}
    \coloneqq 
    \begin{bmatrix}
        \bPhi_{\bm{\eta}}(\bx) \\ \bPhi_{\bz}(\bx)
    \end{bmatrix}
    =
    \bPhi(\bx),
\end{equation}
where $\bm{\eta}\in\mc{N}$, $\dim(\mc{N})\eqqcolon N$, captures the actuated components of the state (e.g., the robot's joint positions and velocities) and $\bz\in\mc{Z}$, $\dim(\mc{Z})\eqqcolon Z$, captures the unactuated components of the state (e.g., the robot's center of mass position and velocity), and $N+Z=2n$. Assuming that $\pdv{\bPhi_{\bz}}{\bx}\bG(\bx)\equiv \bzero$, this coordinate transformation decomposes \eqref{eq:fom_HC} into actuated and unactuated dynamics:
\begin{equation} \label{eq:fom_HC_normal_form}
        \mathcal{HC}
        =
        \left\{
        \begin{array}{ll}
              \dot {\bm{\eta}} = \hat{\bf}(\bm{\eta},\bz) + \hat{\bg}(\bm{\eta},\bz) \bu & \text{if} \quad \bPhi^{-1}(\bm{\eta},\bz) \in \mathcal{X} \backslash \mc{S}\\
              \dot{\bz} = \bm{\omega}(\bm{\eta},\bz) & \text{if} \quad \bPhi^{-1}(\bm{\eta},\bz) \in \mathcal{X} \backslash \mc{S}\\
              \bm{\eta}^{+} = \bm{\Delta}_{\bm{\eta}}(\bm{\eta}^{-}, \bz^{-})    & \text{if} \quad \bPhi^{-1}(\bm{\eta}^-,\bz^-) \in \mc{S} \\
               \bm{\bz}^{+} = \bm{\Delta}_{\bm{\bz}}(\bm{\eta}^{-}, \bz^{-})    & \text{if} \quad \bPhi^{-1}(\bm{\eta}^-,\bz^-) \in \mc{S}.
        \end{array} 
        \right. 
\end{equation}
A common coordinate transformation for \eqref{eq:fom_continuous_dyn} that achieves this decomposition is $\bPhi(\bx)\coloneqq(\bPhi_{\bm{\eta}}(\bx),\bPhi_{\bz}(\bx))$ where:
\begin{align}
    \bm{\eta} =
    \begin{bmatrix}
        \boldsymbol \eta_1 \\
        \boldsymbol \eta_2
    \end{bmatrix}
    \coloneqq
    \mathbf{\Phi}_{\bm{\eta}} (\bx) \coloneqq &
    \begin{bmatrix}
        \bB^{\top} \bq \\
        \bB^{\top} \dot{\bq}
    \end{bmatrix}
    , \label{eq:eta_coordinates}
    \\
    \bz =
    \begin{bmatrix}
        \bz_1 \\
        \bz_2
    \end{bmatrix}
    \coloneqq 
    \mathbf{\Phi}_{\bz} (\bx) \coloneqq &
    \begin{bmatrix}
        \bN \bq \\
        \bN \bD(\bq) \dot{\bq}
    \end{bmatrix},
    \label{eq:z_coordinates}
\end{align}
where
$\bN$ is a left annihilator of $\bB$ in the sense that $\bN\bB=\bzero$.

Given a locally Lipschitz feedback controller $\hat{\bk}\,:\,\mc{N}\times\mc{Z}\rightarrow\mc{U}$, we may also consider the closed-loop hybrid system $\mc{H}$ by taking $\bu=\hat{\bk}(\bm{\eta},\bz)$ in \eqref{eq:fom_HC_normal_form}, which may be used to define a corresponding Poincar\'e map as in \eqref{eq:fom_discrete_dyn} to study stability properties of the corresponding hybrid system.

\subsection{Lyapunov Stability} \label{sec:lyapunov_background}

Powerful statements regarding stability of hybrid systems can be made based purely on properties of the Poincar\'e map \eqref{eq:fom_discrete_dyn}, a discrete-time dynamical system. We now review notions of stability for discrete-time dynamical systems. 
\begin{definition}
    A fixed-point $\bx^*=\bP(\bx^*)$ of \eqref{eq:fom_discrete_dyn} is said to be locally exponentially stable if there exist $M>0$, $\alpha\in(0,1)$, and $\delta>0$ such that for all $\bx_0\in\mc{B}_{\delta}(\bx^*)\cap\mc{S}$:
    \begin{equation}
        \|\bx_{k} - \bx^*\| \leq M\alpha^{k}\|\bx_0 - \bx^*\|,\quad \forall k\in\mathbb{N}.
    \end{equation}
\end{definition}
In this paper, we leverage approximations of the Poincar\'e map to make statements regarding stability of the full-order dynamics. Here, the approximation error  -- which, as discussed later, represents the difference between the ``true" reduced-order dynamics and the ``desired" reduced-order dynamics encoded by a ROM -- may be viewed as a disturbance input, yielding the perturbed dynamics:
\begin{equation}\label{eq:perturbed-poincare-map}
    \bx_{k+1} = \bP(\bx_k,\bd_k),
\end{equation}
where $\bd\in\mc{D}$ is a disturbance input. More details on perturbed Poincar\'e maps can be found in \cite{PoulakakisTAC19}. The notion of input-to-state stability (ISS) provides a framework for analyzing the properties of such perturbed systems \cite{JiangAutomatica01,JiangSCL02,PoulakakisTAC19}.
\begin{definition}[\cite{PoulakakisTAC19}]
    The perturbed Poincar\'e map \eqref{eq:perturbed-poincare-map} is said to be locally exponentially input-to-state stable (E-ISS) if there exist $M>0$, $\alpha\in(0,1)$, $\gamma\in\mc{K}$, and $\delta>0$ such that $\|\bd\|_{\infty}\leq \delta$ and for all $\bx_0\in\mc{B}_{\delta}(\bx^*)\cap\mc{S}$:
    \begin{equation}\label{eq:E-ISS}
        \|\bx_{k} - \bx^*\| \leq M\alpha^{k}\|\bx_0 - \bx^*\| + \gamma\left(\|\bd\|_{\infty}\right),\quad \forall k\in\mathbb{N},
    \end{equation}
    where $\|\bd\|_{\infty}\coloneqq \sup_{k\in\mathbb{N}}\|\bd_k\|$.
\end{definition}
\begin{definition}[\cite{PoulakakisTAC19}]
    A continuous function $V\,:\,\mc{S}\rightarrow\R_{\geq0}$ is said to be a local E-ISS Lyapunov function for \eqref{eq:perturbed-poincare-map} if there exist $c_1,c_2,\delta>0$, $c_3\in(0,1)$, and $\sigma\in\mc{K}$ such that for all $\bx\in\mc{B}_{\delta}(\bx^*)\cap\mc{S}$ and all $\|\bd\|\leq\delta$:
    \begin{equation}
        c_1\|\bx - \bx^*\|^2 \leq V(\bx) \leq c_2 \|\bx - \bx^*\|^2
    \end{equation}
    \begin{equation}
        V(\bP(\bx,\bd)) - V(\bx) \leq -c_3V(\bx) + \sigma(\|\bd\|).
    \end{equation}
\end{definition}
The existence of E-ISS Lyapunov functions are necessary and sufficient
for E-ISS of discrete-time systems \cite{JiangAutomatica01}. While we will leverage the notion of E-ISS to characterize the discrete-time dynamics of hybrid systems via the Poincar\'e map, we will use the notion of a \emph{rapidly} exponentially stabilizing control Lyapunov function \cite{ames2014rapidly} to ensure that continuous stability is preserved under impacts. 
\begin{definition}[\cite{ames2014rapidly}]
    A continuously differentiable function $V_{\varepsilon}\,:\,\mc{X}\rightarrow\R_{\geq0}$ is said to be a \emph{rapidly exponentially stabilizing control Lyapunov function} (RES-CLF) for \eqref{eq:fom_continuous_dyn} with respect to a set $\mc{M}\subset\mc{X}$ if there exist constants $c_1,c_2,c_3>0$ and $\varepsilon\in(0,1)$ such that for all $\bx\in\mc{X}$:
    \begin{equation}
        c_1\|\bx\|_{\mc{M}}^2 \leq V_{\varepsilon}(\bx) \leq \frac{c_2}{\varepsilon^2}\|\bx\|_{\mc{M}}^2
    \end{equation}
    \begin{equation}\label{eq:RES-CLF}
        \inf_{\bu\in\mc{U}}\left\{\pdv{V_{\varepsilon}}{\bx}(\bx)\bF(\bx) + \pdv{V_{\varepsilon}}{\bx}(\bx)\bG(\bx)\bu \right\} \leq - \frac{c_3}{\varepsilon}V_{\varepsilon}(\bx).
    \end{equation}
\end{definition}
As shown in \cite{ames2014rapidly}, any controller satisfying the condition in \eqref{eq:RES-CLF} exponentially stabilizes the continuous dynamics to $\mc{M}$ at a rate determined by $\varepsilon$. In the context of bipedal locomotion, RES-CLFs can often be synthesized using input-output linearization \cite{ames2014rapidly}.

%% file: Sections/HZDROM.tex
\section{From Hybrid Zero Dynamics to Reduced-Order Models} 
\label{sec:HZD_to_ROM}
\subsection{Hybrid Zero Dynamics}
\label{sec:HZD}
The main objective of this paper is to stabilize a full-order hybrid system \eqref{eq:fom_HC_normal_form} by relating desired behaviors of the unactuated dynamics to those of the actuated dynamics via ROMs. To this end, consider the set:
\begin{equation}\label{eq:zero-dyn-manifold}
    \begin{aligned}
        \mc{M} \coloneqq & \{(\bm{\eta}, \bz)\in\mc{N}\times\mc{Z}\,|\,\bh(\bm{\eta},\bz) = \bm{\eta}-\bpsi(\bz)= \bzero\}, \\ 
    \end{aligned}
\end{equation}
where $\bpsi\,:\,\mc{Z}\rightarrow\mc{N}$ is a continuously differentiable mapping relating the underactuated states to the actuated states. As noted in \cite{csomay2024robust,compton2024constructive}, this paradigm of relating unactuated states to actuated ones is central in locomotion techniques that leverage ROMs. Since $\bh\,:\,\mc{N}\times\mc{Z}\rightarrow\R^N$ and 
$[\pdv{\bh}{\bm{\eta}}\;-\pdv{\bh}{\bz}]=[\bI_N\;-\pdv{\bpsi}{\bz}(\bz)]$ 
has rank $N$, $\mc{M}$ is a differentiable manifold of dimension $2n - N=Z$. To relate the behaviors of \eqref{eq:fom_HC_normal_form} to a ROM, we must restrict the evolution of this hybrid system to a lower-dimensional surface. The following result shows that for the decomposition in \eqref{eq:eta_coordinates} and \eqref{eq:z_coordinates} and particular classes $\bpsi$, it is possible to render \eqref{eq:zero-dyn-manifold} invariant.
\begin{lemma}\label{lemma:continuous-invariance}
    Consider the hybrid control system \eqref{eq:fom_HC_normal_form} obtained using the actuation decomposition from \eqref{eq:eta_coordinates} and \eqref{eq:z_coordinates}. Define $\bpsi(\bz)\coloneqq (\bpsi_{1}(\bz),\bpsi_{2}(\bz))$, where $\bpsi_2(\bz)\coloneqq\pdv{\bpsi_1}{\bz}(\bz)\bm{\omega}(\bm{\eta},\bz)$ for all $\bm{\eta}\in\mc{N}$ is independent of $\bm{\eta}$. Provided, $\rank(\bB)=m$ and $\mc{U}=\R^m$, the manifold $\mc{M}$ defined as in \eqref{eq:zero-dyn-manifold} is controlled invariant for the continuous dynamics of \eqref{eq:fom_HC_normal_form}.
    In particular, the controller $\hat{\bk}(\bm{\eta},\bz)\coloneqq \bk(\bPhi^{-1}(\bm{\eta},\bz))$, where:
    \begin{equation}\label{eq:u-M-invariant}
    \begin{aligned}
        \bk(\bx) \coloneqq & (\bB\T\bD^{-1}(\bq)\bB)^{-1}\big[\bB\T\bD^{-1}(\bq)\bH(\bx)  + \dot{\bm{\psi}}_2(\bPhi(\bx))  \big]
    \end{aligned}
    \end{equation}
    renders $\mc{M}$ invariant for the continuous dynamics of \eqref{eq:fom_HC_normal_form}.
\end{lemma}
\begin{proof}
    The set $\mc{M}$ is invariant for the continuous portion of the hybrid dynamics if and only if for each $(\bm{\eta},\bz)\in\mc{M}$ there exists an input $\bu\in\mc{U}$ such that:
    \begin{equation}\label{eq:TM}
        \begin{bmatrix}
            \hat{\bf}(\bm{\eta},\bz) + \hat{\bg}(\bm{\eta},\bz) \bu \\ \bm{\omega}(\bm{\eta},\bz)
        \end{bmatrix}
        \in \mathsf{T}_{(\bm{\eta},\bz)}\mc{M}.
    \end{equation}
    Since $\mc{M}$ is defined as the zero level set of a differentiable function, \eqref{eq:TM} is equivalent to the requirement that for each $(\bm{\eta},\bz)\in\mc{M}$ there exists an input $\bu\in\mc{U}$ such that $\dot{\bh}(\bm{\eta},\bz,\bu)=\bzero$. Given the form of $\bpsi$ and $\bm{\eta}$, we have:
    \begin{equation*}
        \begin{aligned}
            \dot{\bh} = &
        \begin{bmatrix}
            \dot{\bm{\eta}}_1 - \pdv{\bpsi_1}{\bz}(\bz)\bm{\omega}(\bm{\eta},\bz) \\ 
            \dot{\bm{\eta}}_2 - \pdv{\bpsi_2}{\bz}(\bz)\bm{\omega}(\bm{\eta},\bz) 
        \end{bmatrix} \\
        = &
        \begin{bmatrix}
            \bm{\eta}_2 - \bpsi_2(\bz) \\ \bB\T\bD^{-1}\bB\bu - \bB\T\bD^{-1}\bH(\bPhi^{-1}(\bm{\eta}, \bz)) - \dot{\bpsi}_2(\bm{\eta},\bz) 
        \end{bmatrix}.
        \end{aligned}
    \end{equation*}
    Further, since $\bm{\eta}=\bpsi(\bz)$ on $\mc{M}$,
    one may verify that, with $\bu=\hat{\bk}(\bm{\eta},\bz)$, which exists since $\rank(\bB)=m$ and $\mc{U}=\R^m$, we have $\dot{\bh}(\bm{\psi}(\bz),\bz,\bu)=\bzero$, 
    implying invariance of $\mc{M}$.
\end{proof}
Although the above result guarantees the existence of inputs that render $\mc{M}$ invariant\footnote{Under the same assumptions as Lemma \ref{lemma:continuous-invariance}, one may also establish the existence of an RES-CLF for $\mc{M}$ using input-output linearization as in \cite{ames2014rapidly}.} for the continuous dynamics in \eqref{eq:fom_HC_normal_form}, discrete transitions may destroy this invariance property for the overall hybrid system. The following result provides conditions on $\bpsi$ that preserve invariance of $\mc{M}$ under resets.
\begin{lemma}\label{lemma:discrete-invariance}
    The manifold $\mc{M}$ from \eqref{eq:zero-dyn-manifold} is invariant for the discrete dynamics of the hybrid system \eqref{eq:fom_HC_normal_form} if for all $\bz\in\mc{Z}$:
    \begin{equation}\label{eq:discrete-invariance}
    \bm{\Delta}_{\bm{\eta}}(\bpsi(\bz), \bz) = \bpsi(\bm{\Delta}_{\bz}(\bpsi(\bz), \bz)).
\end{equation}
\end{lemma}
\begin{proof}
    To show discrete invariance, we must have $(\bm{\Delta}_{\bm{\eta}}(\bm{\eta},\bz),\bm{\Delta}_{\bz}(\bm{\eta},\bz))\in\mc{M}$
    for all $(\bm{\eta},\bz)\in\mc{M}$, which, based on \eqref{eq:zero-dyn-manifold} is equivalent to:
    \begin{equation}\label{eq:discrete-invariance-2}
        \bm{\Delta}_{\bm{\eta}}(\bm{\eta},\bz) = \bm{\psi}(\bm{\Delta}_{\bz}(\bm{\eta},\bz)),\quad (\bm{\eta},\bz)\in\mc{M}.
    \end{equation}
    Since $\bm{\eta}=\bpsi(\bz)$ on $\mc{M}$, \eqref{eq:discrete-invariance} implies that \eqref{eq:discrete-invariance-2} holds.
\end{proof}
Putting Lemma \ref{lemma:continuous-invariance} and Lemma \ref{lemma:discrete-invariance} together, we have the following result.
\begin{theorem}
    Let the conditions of Lemma \ref{lemma:continuous-invariance} and Lemma \ref{lemma:discrete-invariance} hold. Then, 
    the feedback controller $\hat{\bk}$ from Lemma \ref{lemma:continuous-invariance}
    renders the manifold $\mc{M}$ from \eqref{eq:zero-dyn-manifold} hybrid invariant for the hybrid closed-loop system. 
\end{theorem}
\begin{proof}
    Hybrid invariance requires trajectories of the hybrid system starting on $\mc{M}$ to remain on $\mc{M}$ for all time. Lemma \ref{lemma:continuous-invariance} ensures that trajectories of the hybrid system remain on $\mc{M}$ during continuous evolution, whereas Lemma \ref{lemma:discrete-invariance} ensures that trajectories remain on $\mc{M}$ after resets, thereby ensuring hybrid invariance, as desired.
\end{proof}
When the manifold $\mc{M}$ is hybrid invariant, the full-order hybrid system evolves on a lower-dimensional surface, $\mc{M}$, governed by a reduced-order hybrid system, the \emph{hybrid zero dynamics} (HZD). To characterize the HZD, we must determine a set of local coordinates for $\mc{M}$, which may be taken as $\bz\in\mc{Z}$ since: 
\begin{equation}\label{eq:iota}
    \bm{\iota}(\bz)\coloneqq(\bpsi(\bz), \bz), 
\end{equation}
where $\bm{\iota}\,:\,\mc{Z}\rightarrow\mc{N}\times\mc{Z}$, is an embedding, provides a parameterization of $\mc{M}$. 
Moreover, $\iota^{-1}(\bm{\eta},\bz)=\bz$ is well-defined on $\mc{M}$, illustrating that $\mc{M}\cong\mc{Z}$.
Hence, the HZD may be expressed as:
\begin{equation} \label{eq:fom_H_zero_dyn}
        \mathcal{H}|_{\mc{M}}
        =
        \left\{
        \begin{array}{ll}
              \dot{\bz} = \bm{\omega}(\bpsi(\bz),\bz) & \text{if} \quad \bz \in \mc{Z} \backslash (\bPhi(\mc{S})\cap\mc{Z})\\
               \bm{\bz}^{+} = \bm{\Delta}_{\bz}(\bpsi(\bz),\bz)   & \text{if} \quad \bz^{-} \in \bPhi(\mc{S})\cap\mc{Z}.
        \end{array} 
        \right.
\end{equation}
Following a similar approach to that in Sec. \ref{sec:fom_dynamics}, \eqref{eq:fom_H_zero_dyn} induces a discrete-time system via the Poincar\'e map:
\begin{equation}\label{eq:fom-discrete-hzd}
    \bz_{k+1} = \bm{\Omega}(\bz_k),\quad \bm{\Omega}(\bz) \coloneqq \bm{\varphi}^{\bz}_{T_{I}^{\bz}(\bz)}\big(\bm{\Delta}_{\bm{\bz}}(\bpsi(\bz), \bz)\big),
\end{equation}
where $\flow^{\bz}_{t}(\bz)$ is the flow of $\dot{\bz} = \bm{\omega}(\bpsi(\bz),\bz)$, $T_{I}^{\bz}(\bz)$ is the time-to-impact function defined similarly to \eqref{eq:fom-time-to-impact}, and $\bm{\Omega}\,:\,\bPhi(\mc{S})\cap\mc{Z}\rightarrow \bPhi(\mc{S})\cap\mc{Z}$ is, again, a partial function. The power of HZD is the ability to infer stability properties of the full-order hybrid system based only on properties of the reduced-order system in \eqref{eq:fom_H_zero_dyn}. In what follows, we demonstrate how such full-order stability properties can be inferred from approximations of \eqref{eq:fom-discrete-hzd}.

\subsection{Reduced-Order Models} \label{sec:roms}
\label{sec:ROM}
Rather than work directly with the HZD \eqref{eq:fom_H_zero_dyn}, the approach taken herein is to leverage approximations of the corresponding Poincar\'e map \eqref{eq:fom-discrete-hzd} while still making practical statements regarding stability of the full-order system.  
To this end, consider the discrete-time system:
\begin{equation}\label{eq:rom-discrete-dynamics}
    \br_{k+1} = \bQ(\br_k,\bm{\ell}_k),
\end{equation}
where $\br\in\mc{R}$ is the state of the ROM with control input $\bm{\ell}\in\mc{L}$, and
$\bQ\,:\,\mc{R}\times\mc{L}\rightarrow\mc{R}$ is viewed as an approximation of the Poincar\'e map in \eqref{eq:fom-discrete-hzd}.
In practice, the states $\bz$ of \eqref{eq:fom_H_zero_dyn} and $\br$ of \eqref{eq:rom-discrete-dynamics} often have the same (up to a diffeomorphism) semantic interpretation (e.g., center of mass position and velocity), the discrete control input $\bm{\ell}$ represents a choice made each stride, such as the step length, and $\bQ$ typically captures the step-to-step dynamics of, e.g., a LIP.
Given a controller $\bm{\pi}\,:\,\mc{R}\rightarrow\mc{L}$ for  \eqref{eq:rom-discrete-dynamics}, taking $\bm{\ell}=\bm{\pi}(\br)$ produces the closed-loop system:
\begin{equation}\label{eq:rom-cl-discrete-dynamics}
    \br_{k+1} = \bQ_{\rm{cl}}(\br_k) \coloneqq \bQ(\br_k,\bm{\pi}(\br_k)).
\end{equation}

\begin{remark}
    As discussed later in Sec. \ref{sec:biped}, the behavior of the above ROM will be encoded via $\bpsi$, and, therefore, in $\mc{M}$. In this regard, $\bpsi$ may be seen as generating desired commands for the actuated components of the full-order hybrid system \eqref{eq:fom_HC_normal_form} (e.g., foot placement commands) based on simplified representations (e.g., inverted pendulum models) of the HZD. 
\end{remark}

Rather than directly transferring stability properties of the discrete-time HZD \eqref{eq:fom-discrete-hzd} back to the full-order hybrid system, we seek to transfer stability properties of the ROM \eqref{eq:rom-cl-discrete-dynamics} back to the full-order hybrid system. To make any meaningful statements in this context, we must first establish a connection between the ROM \eqref{eq:rom-discrete-dynamics} and the HZD \eqref{eq:fom-discrete-hzd}.
\begin{lemma}\label{lemma:zero-dyn-orbit}
    Consider the HZD \eqref{eq:fom_H_zero_dyn}, the discrete-time ROM \eqref{eq:rom-cl-discrete-dynamics}, and let $\br^*=\bQ_{\rm{cl}}(\br^*)$ be a fixed point of \eqref{eq:rom-cl-discrete-dynamics}. 
    Provided that $\mc{R}\cong\mc{Z}$ for a diffeomorphism $\bm{\Xi}\,:\,\mc{R}\rightarrow\mc{Z}$ and $\bz^*\coloneqq\bXi(\br^*)$ is a fixed point for $\bm{\Omega}(\bz^*)=\bz^*$, then:
    \begin{equation}\label{eq:Oz}
        \mc{O}_{\rm{z}} \coloneqq \{\bm{\varphi}^{\bz}_{t}\big(\bm{\Delta}_{\bm{\bz}}(\bpsi(\bz^*), \bz^*)\big)\,|\,t\in[0,T_{I}^{\bz}(\bz^*)]\},
    \end{equation}
    is a periodic orbit of hybrid zero dynamics \eqref{eq:fom_H_zero_dyn}.
\end{lemma}
\begin{proof}
    If $\bz^*$ is a fixed point of $\bm{\Omega}$, then:
    \begin{equation*}
        \bz^* = \bm{\Omega}(\bz^*) = \bm{\varphi}^{\bz}_{T_{I}^{\bz}(\bz^*)}\big(\bm{\Delta}_{\bm{\bz}}(\bpsi(\bz^*), \bz^*)\big),
    \end{equation*}
    which implies that the flow $\flow_t^\bz(\bz)$ is periodic with period $T_{\bz^*}=T_{I}^{\bz}(\bz^*)$. Since the flow is periodic with respect to $\bz^*$, $\mc{O}_{\bz}$ in \eqref{eq:Oz} is a periodic orbit of \eqref{eq:fom_H_zero_dyn}, as desired.
\end{proof}
While the above result establishes the connection between a ROM and a periodic orbit for the HZD, nothing, as of yet, allows for making statements regarding the stability of such an orbit. 
Before proceeding, we define:
\begin{equation}\label{eq:d}
    \bd_{k} \coloneqq \bm{\Omega}(\bz_k) - \bm{\Xi}\circ\bQ_{\rm{cl}}\circ\bm{\Xi}^{-1}(\bz_k),
\end{equation}
as the discrepancy between the Poincar\'e map of the HZD \eqref{eq:fom-discrete-hzd} and the ROM \eqref{eq:rom-cl-discrete-dynamics} represented in $\bz$ coordinates.
The following theorem captures the main result of this paper and establishes when stability properties of the ROM may be transferred back to the original \emph{full-order} hybrid system.

\begin{theorem}\label{theorem:main}
    Consider the hybrid system \eqref{eq:fom_HC_normal_form}, the manifold $\mc{M}$ from \eqref{eq:zero-dyn-manifold}, and suppose there exists an RES-CLF $V_{\varepsilon}\,:\,\mc{N}\times\mc{Z}\rightarrow\R_{\geq0}$ for \eqref{eq:fom_HC_normal_form} with respect to $\mc{M}$. Let $\hat{\bk}_{\varepsilon}\,:\,\mc{N}\times\mc{Z}\rightarrow\mc{U}$ be any locally Lipschitz feedback controller satisfying the RES-CLF condition \eqref{eq:RES-CLF} and suppose that $\mc{M}$ is hybrid invariant under this controller with corresponding HZD in \eqref{eq:fom_H_zero_dyn}. Let the conditions of Lemma \ref{lemma:zero-dyn-orbit} hold and suppose that  $\br^*=\bQ_{\rm{cl}}(\br^*)$ is an exponentially stable fixed point of the corresponding ROM \eqref{eq:rom-cl-discrete-dynamics}. Then, there exists $\overline{\varepsilon}>0$ 
    and $\rho>0$
    such that for all $\varepsilon\in(0,\overline{\varepsilon})$
    and all $\bd_k$ satisfying $\|\bd\|_{\infty}\leq\rho$,
    $\mc{O}\coloneqq \bm{\iota}(\mc{O}_{\bz})$, with $\iota$ as in \eqref{eq:iota}, is a locally E-ISS periodic orbit of the closed-loop hybrid system. 
\end{theorem}

Before proving Theorem \ref{theorem:main}, we present a key technical lemma, 
illustrating that if the discrete-time ROM \eqref{eq:rom-cl-discrete-dynamics} is exponentially stable and its state space is diffeomorphic to that in \eqref{eq:fom-discrete-hzd}, then the discrete-time HZD \eqref{eq:fom-discrete-hzd} are E-ISS.
\begin{lemma}\label{lemma:discrete-time-iss}
    Let the conditions of Lemma \ref{lemma:zero-dyn-orbit} hold and suppose that $\br^*=\bQ_{\rm{cl}}(\br^*)$ is an exponentially stable fixed point of 
    \eqref{eq:rom-cl-discrete-dynamics}. 
    Then, there exists a $\rho>0$ such that for all $\bd_{k}$ satisfying $\|\bd\|_{\infty}\leq\rho$, the discrete-time HZD \eqref{eq:fom-discrete-hzd} are locally E-ISS. 
\end{lemma}

\begin{proof}
    Since $\br^*$ is an exponentially stable fixed point of \eqref{eq:rom-cl-discrete-dynamics}, there exist $M>0$, $\alpha\in(0,1)$, and $\delta>0$ such that for all 
    $\br_0\in\mc{B}_{\delta}(\br^*)$ 
    we have:
    \begin{equation*}
        \|\br_k - \br^*\| \leq M\alpha^k\|\br_0 - \br^*\| \leq M \|\br_0 - \br^*\|  \eqqcolon\epsilon, \quad \forall k\in\mathbb{N},
    \end{equation*}
    which implies that $\br_k\in\overline{\mc{B}_{\epsilon}(\br^*)}$
    for all $k\in\mathbb{N}$. Let $\bz=\bm{\Xi}(\br)$ and define:
    \begin{equation}\label{eq:rom-cl-z-dynamics}
        \bz_{k+1} = \bQ_{\bz}(\bz_k),\quad \bQ_{\bz}(\bz) \coloneqq \bm{\Xi}\circ\bQ_{\rm{cl}}\circ\bm{\Xi}^{-1}(\bz),
    \end{equation}
    as the closed-loop ROM dynamics \eqref{eq:rom-cl-discrete-dynamics} in the $\bz$ coordinates with fixed point $\bz^*=\bm{\Xi}(\br^*)$. Since $\br_k\in\overline{\mc{B}_{\epsilon}(\br^*)}$, we have $\bz_{k}\in \bXi(\overline{\mc{B}_{\epsilon}(\bXi^{-1}(\bz^*))})$ for all $k\in\mathbb{N}$.  As $\bm{\Xi}$ is a diffeomorphism, it is locally Lipschitz continuous and hence Lipschitz continuous on any compact set. Thus, there exists an $\overline{L}>0$ such that for all $\br\in\overline{\mc{B}_{\epsilon}(\br^*)}$ and all $\bz\in \bXi(\overline{\mc{B}_{\epsilon}(\bXi^{-1}(\bz^*))})$:
    \begin{equation*}
        \|\bz - \bz^*\| = \|\bXi(\br) - \bXi(\br^*)\|  \leq \overline{L}\|\br - \br^*\|.
    \end{equation*}
    Similarly, $\bXi^{-1}$ is also locally Lipschitz, implying the existence of an $\underline{L}>0$ such that for all $\br\in\overline{\mc{B}_{\epsilon}(\br^*)}$ and all $\bz\in\bXi(\overline{\mc{B}_{\epsilon}(\bXi^{-1}(\bz^*))})$:
    \begin{equation*}
        \|\br - \br^*\| = \|\bXi^{-1}(\bz) - \bXi^{-1}(\bz^*)\|  \leq \underline{L}\|\bz - \bz^*\|.
    \end{equation*}
    Using the preceding bounds by noting that $\br_k\in\overline{\mc{B}_{\epsilon}(\br^*)}$ and $\bz_k\in\bXi(\overline{\mc{B}_{\epsilon}(\bXi^{-1}(\bz^*))})$ for all $k\in\mathbb{N}$, we have:
    \begin{equation*}
        \begin{aligned}
            \|\bz_k - \bz^*\| \leq  \overline{L}\|\br_k - \br^*\| \leq & \overline{L}M\alpha^k\|\br_0 - \br^*\| \\ 
            \leq & \underline{L}\overline{L}M\alpha^k\|\bz_0 - \bz^*\|,
        \end{aligned}
    \end{equation*}
    which demonstrates local exponential stability of $\bz^*$ for \eqref{eq:rom-cl-z-dynamics}. As $\bz^*$ is locally exponentially stable for \eqref{eq:rom-cl-z-dynamics}, the converse Lyapunov theorems in \cite[Thm. 1]{JiangSCL02} and \cite[Thm. 2]{JiangSCL02} guarantee the existence of a continuously differentiable function $V_{\bz}\,:\,\bPhi(\mc{S})\cap\mc{Z}\rightarrow\R_{\geq0}$ satisfying:
    \begin{equation}
        c_1\|\bz - \bz^*\|^2 \leq V_{\bz}(\bz) \leq c_2\|\bz - \bz^*\|^2,
    \end{equation}
    \begin{equation}
        V_{\bz}(\bQ_{\bz}(\bz)) - V_{\bz}(\bz) \leq -c_3V_{\bz}(\bz),
    \end{equation}
    for all $\bz\in\mc{B}_{\rho}(\bz^*)\cap\bPhi(\mc{S})\cap\mc{Z}$ with $c_1,c_2,\rho>0$ and $c_3\in(0,1)$.
    To show that this implies E-ISS of the discrete-time hybrid zero dynamics, we represent 
    \eqref{eq:fom-discrete-hzd} as:
    \begin{equation*}
        \bz_{k+1} = \bm{\Omega}(\bz_{k}) \pm \bQ_{\bz}(\bz_{k}) = \bQ_{\bz}(\bz_{k}) + \bd_{k},
    \end{equation*}
    where $\bd_{k}= \bm{\Omega}(\bz_{k}) - \bQ_{\bz}(\bz_{k})$.
    Next, we note that:
    \begin{equation*}
        \begin{aligned}
            V_{\bz}(\bOmega(\bz)) = & V_{\bz}(\bQ_{\bz}(\bz) + \bd) \\ 
            = & V_{\bz}(\bQ_{\bz}(\bz) + \bd) \pm V_{\bz}(\bQ_{\bz}(\bz)) \\ 
            \leq & V_{\bz}(\bQ_{\bz}(\bz))  + |V_{\bz}(\bQ_{\bz}(\bz) + \bd) - V_{\bz}(\bQ_{\bz}(\bz))| \\
            \leq & V_{\bz}(\bQ_{\bz}(\bz)) + L\|\bd\|,
        \end{aligned}
    \end{equation*}
    where the first inequality follows from the triangle inequality and the second from the Lipschitz continuity of $V_{\bz}$ with Lipschitz constant $L>0$. The above implies that:
    \begin{equation}
        \begin{aligned}
            V_{\bz}(\bm{\Omega}(\bz)) - V_{\bz}(\bz) \leq & V_{\bz}(\bQ_{\bz}(\bz)) - V_{\bz}(\bz) + L\|\bd\| \\
            \leq & -c_3V_{\bz}(\bz) + L\|\bd\|, \\
        \end{aligned}
    \end{equation}
    for all $\bz\in\mc{B}_{\rho}(\bz^*)\cap\bPhi(\mc{S})\cap\mc{Z}$ and any $\bd$ 
    such that
    $\|\bd\|_{\infty}\leq \rho$. Hence, $V_{\bz}$ is a local E-ISS Lyapunov function for \eqref{eq:fom-discrete-hzd}, which implies that \eqref{eq:fom-discrete-hzd} is E-ISS.
\end{proof}

With the preceding lemma, we now have the tools in place to establish the main result of this paper.

\begin{proof}[Proof (of Theorem \ref{theorem:main})]
    Lemma \ref{lemma:discrete-time-iss} ensures that, under the conditions of the theorem statement, the Poincar\'e map $\bm{\Omega}\,:\,\bPhi(\mc{S})\cap\mc{Z}\rightarrow\bPhi(\mc{S})\cap\mc{Z}$ associated with the periodic orbit $\mc{O}_{\bz}$ of the HZD \eqref{eq:fom_H_zero_dyn} is locally E-ISS. Given that the Poincar\'e map of the HZD is locally E-ISS by Lemma \ref{lemma:discrete-time-iss}, it follows directly from \cite[Thm. 1]{PoulakakisTAC19} that the corresponding periodic orbit $\mc{O}_{\bz}$ is also locally E-ISS. Next, note that since $\mc{M}$ is hybrid invariant under the RES-CLF controller $\bu=\hat{\bk}_{\varepsilon}(\bm{\eta},\bz)$, the hybrid periodic orbit $\mc{O}_{\bz}$ of the HZD induces a hybrid periodic orbit $\mc{O}=\bm{\iota}(\mc{O}_{\bz})$ for the full-order system via the embedding $\bm{\iota}(\bz)=(\bpsi(\bz),\bz)$. Since $\mc{O}_{\bz}$ is E-ISS and $V_{\varepsilon}$ is an RES-CLF with respect to $\mc{M}$, \cite[Thm. 1]{ShishirACC18} guarantees the existence of an $\overline{\varepsilon}>0$ such that for all $\varepsilon\in(0,\overline{\varepsilon})$, the RES-CLF controller $\bu=\hat{\bk}_{\varepsilon}(\bm{\eta},\bz)$ enforces local E-ISS of $\mc{O}$ for the closed-loop full-order hybrid system \eqref{eq:fom_HC_normal_form}, as desired. 
\end{proof}

\begin{remark}
    Theorem \ref{theorem:main} establishes E-ISS of the periodic orbit $\mc{O}= \bm{\iota}(\mc{O}_{\bz})$ for the full-order hybrid system \eqref{eq:fom_HC_normal_form}, where 
    \eqref{eq:d}
    is viewed as a disturbance input. If a ROM provides a poor approximation of the HZD, then 
    the conditions of Theorem \ref{theorem:main} may not hold.
    In general, completely characterizing the discrepancy between the HZD and a ROM is challenging; however, we demonstrate empirically in Sec. \ref{sec:biped} that, in the context of bipedal locomotion, ROMs commonly used in practice (e.g., variants of the LIP) lead to a small discrepancy, resulting in stability despite using dramatically simplified representations of the full-order hybrid system. 
\end{remark}

%% file: Sections/BipedApplication.tex
\section{Application to Bipedal Walking}\label{sec:biped}
In this section, we elucidate the developed theory by demonstrating its application to bipedal locomotion, where
we use a hybrid linear inverted pendulum (HLIP) controller to realize walking on a five-link planar biped (Fig. \ref{fig:fom_and_rom}).
%
%
\subsection{Five-Link Biped}
\begin{figure}[t]
    \centering
    \subfloat[\centering]{{\includegraphics[width=0.32\linewidth]{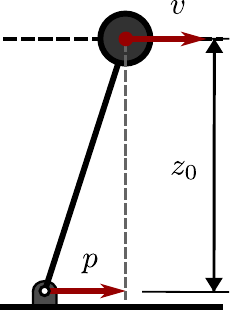} }}%
    \hspace{0.05\linewidth}
    \subfloat[\centering]{{\includegraphics[width=0.32\linewidth]{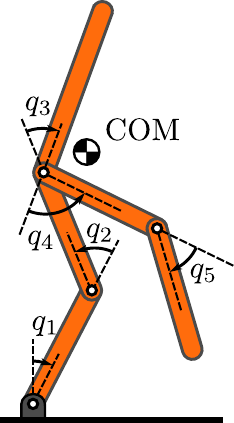} }}%
    \vspace{-0.25cm}
    \caption{(a) A linear inverted pendulum model is used as the ROM for (b) a 5-link planar walker as the FOM. In the FOM, counterclockwise rotation corresponds to positive angles.}%
    \label{fig:fom_and_rom}
    \vspace{-6.0mm}
\end{figure}
Using a \textit{pinned model}, the configuration of the five-link biped is given by $\bq = [q_1 \ q_2 \ q_3 \ q_4 \ q_5]^\top \in \R^5$.
%
%
The dynamics can be written in the form shown in \eqref{eq:fom_continuous_dyn} 
with actuation matrix:
\begin{equation}
    \bB = 
    \begin{bmatrix}
        \bm{0}_{1\times4} \\ \bI_{4\times4}
    \end{bmatrix}
    \in \R^{5 \times 4}.
\end{equation}
In this model, the hips and knees are actuated, $\bq_\mathfrak{a} = \left[ q_2 \;\; q_3 \;\; q_4 \;\; q_5\right]^\top$, and the angle at the pivot point, $q_\mathfrak{u} = q_1$, is unactuated. The reset map for this system is characterized by the instantaneous momentum transfer equations. For a complete discussion on the reset map, see \cite{westervelt2018feedback}.

As discussed in Section \ref{sec:fom_diffeomorphism} we choose a diffeomorphism to transform the state into normal form. We select the actuated and unactuated coordinates as follows:
\begin{align*}
    \bm{\eta} = \mathbf{\Phi}_{\bm{\eta}} (\bx) =&
    \begin{bmatrix}
        \bB^{\top} \bq \\
        \bB^{\top} \dot{\bq}
    \end{bmatrix}
    =
    \begin{bmatrix}
        \bq_\mathfrak{a} \\
        \dot{\bq}_\mathfrak{a}
    \end{bmatrix}
    =
    \begin{bmatrix}
        \bm{\eta}_1 \\ \bm{\eta}_2
    \end{bmatrix}
    \in \R^8
    ,
    \\
    \bz = \mathbf{\Phi}_{\bz} (\bx) = &
    \begin{bmatrix}
        \bN \bq \\
        \bN \bD(\bq) \dot{\bq}
    \end{bmatrix}
    =
    \begin{bmatrix}
        q_\mathfrak{u} \\
        \bD_\mathfrak{u}(\bq) \dot{\bq}
    \end{bmatrix}
    =
    \begin{bmatrix}
        z_1 \\ z_2
    \end{bmatrix}
    \in \R^2.
\end{align*}
%
Here, we select $\bN = [1 \;\;\bzero_{1\times4}]$, thus satisfying $\pdv{\bPhi_{\bz}}{\bx}\bG(\bx)\equiv \bzero$. The unactuated coordinates are then the first configuration state of the position, $q_1$, and the unactuated momentum state. 

\subsection{Hybrid Linear Inverted Pendulum}
The HLIP model is a point mass model with mass $m$ at a fixed height, $z_0 \in \mathbb{R}_{>0}$, and where the state is composed of the mass horizontal position, $p \in \mathbb{R}$, and velocity, $v \in \mathbb{R}$, both relative to the current stance foot \cite{xiong20223}. By design, this reduced-order model captures the ``falling over" nature of the robot. The continuous dynamics for this ROM are:
\begin{equation} \label{eq:hlip_continuous_dyn}
    \dot{\br} = 
    \bA_{\rm{ssp}}\br = 
    \begin{bmatrix}
        0 & 1 \\ \frac{g}{z_0} & 0
    \end{bmatrix}
    \begin{bmatrix}
        p \\ v
    \end{bmatrix}
\end{equation}
where $g$ is the acceleration due to gravity. The reset map for this model occurs at a fixed continuous phase intervals, $T_{\rm{ssp}}$ (single support phase). The discrete dynamics are referred to as \textit{step-to-step (S2S) dynamics}, and are captured as in \eqref{eq:rom-discrete-dynamics}:
%
\begin{equation*}
    \br_{k+1} = 
    e^{\bA_{\rm{ssp}} T_{\rm{ssp}}}
    \br_k
    +
    e^{\bA_{\rm{ssp}} T_{\rm{ssp}}}
    \begin{bmatrix}
        -1 \\
        0 
    \end{bmatrix}
    {\ell}_k
    = \bA_{\rm{R}} \br_k + \bB_{\rm{R}} {\ell}_k.
\end{equation*}
%
For the HLIP model, the ROM dynamics are stabilized by modifying the step length, $\ell_{k} \in \R$. Defining a controller ${\ell}_k = \bm{\pi}(\br_k, c)$ to achieve some velocity command, $c \in \R$, produces a closed loop system:
\begin{equation}
    \br_{k+1} = \bQ_{\rm{cl}}(\br_k) =  \bA_{\rm{R}} \br_{k} + \bB_{\rm{R}}\bm{\pi}(\br_k,c),
\end{equation}%
whose dynamics induce a fixed point $\br^* \in \mc{R}$. 
In this case, $\br^*$ can be rendered exponentially stable through the design of $\bm{\pi}(\br_k, c)$ as in \cite{xiong20223}. These discrete-time dynamics are the Poincar\'e map associated with \eqref{eq:hlip_continuous_dyn} which become a hybrid system based on changing the step size, and thus induce a hybrid periodic orbit, $\mc{O}_{\br}$.
%
%
\subsection{HLIP Embedding onto Five-Link Biped}
As mentioned in Section \ref{sec:roms}, we first define a diffeomorphism $\bm{\Xi}:\mc{R} \rightarrow \mc{Z}$ to compare the HLIP state and biped's underactuated coordinates. We transform the position and velocity coordinates $(p, v)$ of the HLIP into pivot position and angular momentum of the biped $(q_{\mathfrak{u}}, \mathfrak{m})$:
\begin{equation} \label{eq:rom_to_fom_diffeomorphism}
    \bm{\Xi} \coloneq 
    \begin{bmatrix}
        \texttt{IK}_{\rm{base}}(p, z_0) \\
        m z_0 v
    \end{bmatrix}
    =
    \begin{bmatrix}
        q_1 \\ \mathfrak{m}
    \end{bmatrix}
    =
    \begin{bmatrix}
        z_1 \\ z_2
    \end{bmatrix}
\end{equation}
where $\texttt{IK}_{\rm{base}}(\cdot, \cdot)$ solves\footnote{In general, inverse kinematics can provide multiple solutions, but in practice this map can be designed to be smooth and return a unique solution, thus validating the diffeomorphism proposed in \eqref{eq:rom_to_fom_diffeomorphism}.} for $q_1$ given any $p$ and a fixed $z_0$ and the second entry computes the unactuated momentum $\mathfrak{m}$ given a velocity $v$.
We transfer the stability properties of the HLIP to the planar biped by first defining an interface between $\bz$ and $\br^*$. We define the control interface as a simple linear feedback control law:
\begin{equation} \label{eq:hlip_step_controller}
    \kappa(z_1) \coloneqq \ell^* + \bK (\bm{\varphi}^{\br}_{\tau(z_1)}(\bm{\Xi}^{-1}(\bz^+)) - \br^*),
\end{equation}
where $\ell^*\in\R$ is a feedforward term for achieving $c$, $\bK \in \R^{1\times2}$ is a gain matrix, and  $\bm{\varphi}^{\br}_{\tau(z_1)}(\bm{\Xi}^{-1}(\bz^+))$ is the pre-impact state of the FOM in ROM coordinates, approximated by the ROM flow \eqref{eq:hlip_continuous_dyn} for a post-impact state, $\bz^+$. Here, the post impact state $\bz^+$ is constant for the duration of the continuous phase and $\tau(z_1)$ is a parameterization of time similar to \cite{ames2017hybrid}.
As discussed in Section \ref{sec:HZD_to_ROM}, the mapping $\bm{\psi}$ relates the underactuated states to the actuated states. In this context, one method to define $\bm{\psi}$ is through inverse kinematics. 
We lift behaviors from the zero dynamics and ROM to desired behaviors for the actuated coordinates via:
\begin{equation}
    \bm{\psi}(\bz)
    \coloneqq
    \texttt{IK}
    \big(
    \kappa(z_1), \dot{\kappa}(\bz), p_z^d, \theta^d
    \big).
\end{equation}
where $\texttt{IK}$ computes a desired value of $\bm{\eta} = (\bq_{\mathfrak{a}}, \dot{\bq}_{\mathfrak{a}})$ given the desired swing foot state, $(\kappa(z_1), \dot{\kappa}(\bz))$, and the desired center of mass height $p_z^{d}\in\R$ and torso angle $\theta^{d}\in\R$. 
This choice of $\bpsi$ induces the manifold $\mc{M}$ in \eqref{eq:zero-dyn-manifold} to which the system is stabilized.
%
%

\begin{figure}[t!] 
    \centering
    \includegraphics[width=0.95\linewidth]{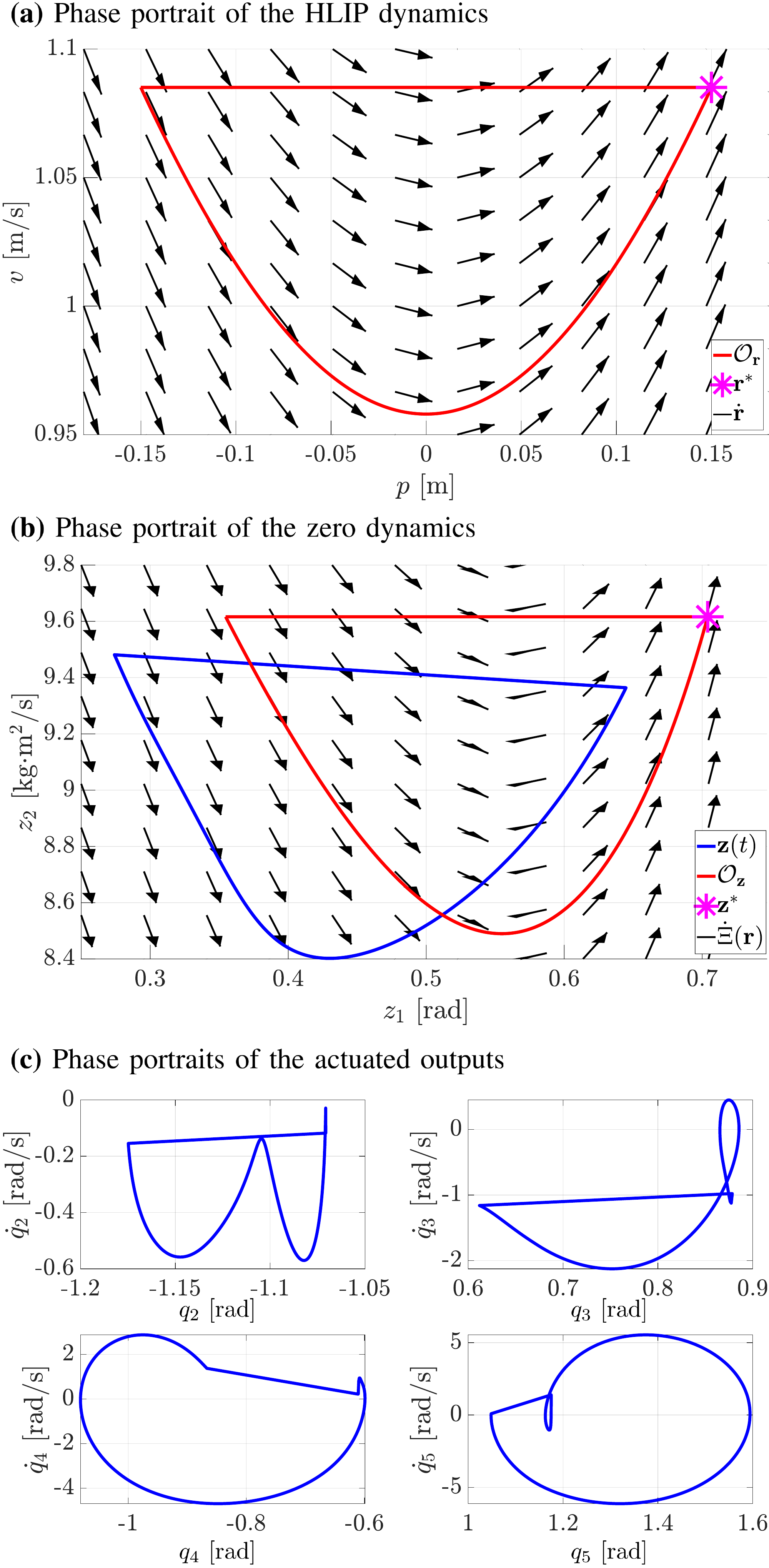}
    \vspace{0.2cm}
    \caption{The plots show steady-state phase portraits over 3 seconds for the robot's zero dynamics and actuated coordinates while walking at 1.0 m/s with a 0.3 s step period. (a) Shows the nominal periodic orbit of the reduced-order model, $\mathcal{O}_{\br}$. (b) Illustrates the ISS zero dynamics of the FOM alongside the transformed ROM dynamics, $\Xi$. (c) Shows the stability of the actuated coordinates.}
    \label{fig:results}
    \vspace{-4mm}
\end{figure}
%
%
\subsection{Implementation Details}
Our implementation differs from our presented theory in a few subtle ways to enhance performance and robustness. 
First, similar to \cite{csomay2024robust}, although our theory relies on an RES-CLF, we leverage PD control to stabilize the system to $\mc{M}$.
Second, we intermittently update $\bz^+$ in \eqref{eq:hlip_step_controller} to the current underactuated state $\bz$, as discussed in \cite{xiong20223,ghansah2024dynamic}. This effectively re-plans the desired swing foot trajectory based on the current state of the robot, rather than using an open-loop trajectory computed once post-impact. While beyond the scope of this paper, accounting for such re-planning in our theory is an important direction for future research.
%
%
\subsection{Analysis}
As shown in Fig. \ref{fig:results}, we stabilize the FOM model by approximating the underactuated dynamics with a ROM. 
Specifically, in Fig. \ref{fig:results} we see that the zero dynamics orbit $\mc{O}_{\bz}$ resembles the ROM orbit $\mc{O}_{\br}$. 
Importantly, while the zero dynamics are not driven to $\mc{O}_{\rm{z}}$, they are driven to a neighborhood of this orbit, reflecting the ISS characterization of Lemma \ref{lemma:discrete-time-iss}.
Since the actuated (Fig. \ref{fig:results}b) and underactuated (Fig. \ref{fig:results}c) coordinates of the robot stabilize to a periodic orbit, via the decomposition, we see that the FOM state is also stabilized to a periodic orbit as guaranteed by Theorem \ref{theorem:main}.

%% file: Sections/Conclusion.tex
\section{Conclusion}\label{sec:conclusion}
This paper presented a formal framework for achieving stable locomotion for hybrid systems based on ROMs. Our key technical results establish connections between the stability of a ROM and that the full-order hybrid system using the notion of input-to-state-stability. Future efforts will be devoted to quantifying the gap between the zero dynamics and commonly used ROMs and leveraging ROMs to formally generate more dynamic walking behaviors.